\documentclass[preprint,12pt]{elsarticle}

\usepackage{amsbsy,amsmath,amsthm,amssymb}
\usepackage{algorithm,algorithmic}
\usepackage{graphics}
\usepackage{fullpage}
\usepackage{url}
\usepackage{multirow}
\usepackage{color}
\usepackage{subcaption}
\usepackage{times}
\usepackage{algorithm}
\usepackage{multirow}
\usepackage{amsmath}
\usepackage{setspace}
\usepackage{xcolor}
\usepackage{soul}
\usepackage{enumitem}

\usepackage[colorlinks,citecolor=blue,urlcolor=blue]{hyperref}
\RequirePackage{hypernat}
\theoremstyle{plain}
\newtheorem{theorem}{Theorem}

\newtheorem{lemma}{Lemma}

\theoremstyle{definition}
\theoremstyle{remark}
\newtheorem{remark}{Remark}

\newtheorem{assumption}{Assumption}
\DeclareMathOperator{\tr}{tr}

\newcommand{\argmin}{\qopname\relax m{arg\,min}}

\newtheorem{assumptionalt}{Assumption}[assumption]
\newenvironment{assumptionp}[1]{
  \renewcommand\theassumptionalt{#1}
  \assumptionalt
}{\endassumptionalt}
\makeatletter
\newcommand{\myitem}[1]{%
\item[#1]\protected@edef\@currentlabel{#1}%
}
\makeatother

\begin{document}

\def\spacingset#1{\renewcommand{\baselinestretch}%
{#1}\small\normalsize} \spacingset{1}

\begin{frontmatter}
\title{Sparse Multivariate Linear Regression with Strongly Associated Response Variables}

\author{Daeyoung Ham\corref{cor1}\fnref{fn1}}
\ead{daeyoung.ham@utsa.edu}
%\affiliation {organization={School of Statistics, University of Minnesota}, addressline={224 Church ST SE},
%postcode={MN 55455}, city={Minneapolis}, country={The United States}}
\author{Bradley S. Price\fnref{fn2}} \ead{
brad.price@mail.wvu.edu}
%\affiliation{organization={Department of Management Information Systems and Supply Chain, West Virginia University}, addressline={83 Beechurst Ave \#4006}, city={Morgantown},
%          postcode={WV 26505},
%          country={The United States}}
\author{Adam J. Rothman\fnref{fn3}} \ead{arothman@umn.edu}
%\affiliation{organization={School of Statistics, University of Minnesota}, addressline={224 Church ST SE},
%postcode={MN 55455}, city={Minneapolis}, country={The United States}}

\cortext[cor1]{Corresponding author}

\begin{abstract}
We propose new methods for multivariate linear regression when the regression coefficient matrix is sparse and the error covariance matrix is dense.
We assume that the error covariance matrix has equicorrelation across the response variables.  
Two procedures are proposed: one is based on constant marginal response variance (compound symmetry), and the other is based on general varying marginal response variance.
Two approximate procedures are also developed for high dimensions.  We propose an approximation to the Gaussian validation likelihood for tuning parameter selection.  Extensive numerical experiments illustrate when our procedures outperform relevant competitors
as well as their robustness to a moderate degree of model misspecification.
\end{abstract}


\end{frontmatter}

\section{Introduction}\label{sec:Intro}

Multivariate linear regression simultaneously models
multiple response variables in terms of 
one or more predictors.  It is well studied, and we direct readers to \citet{reinsel1998multivariate} for 
a detailed review.

We focus on fitting multivariate linear regression
models by penalized or constrained Gaussian likelihood. One of the first approaches maximizes the Gaussian likelihood
subject to a rank constraint on the regression
coefficient matrix \citep{izenman1975reduced, reinsel1998multivariate}.
Like other dimension 
reduction methods, interpreting the fitted model in terms of the original predictors may be difficult.

\citet{MRCE} proposed to jointly estimate the error precision matrix and the regression coefficient matrix by penalized Gaussian likelihood.
%In \citet{MRCE}, the inverse covariance matrix (precision matrix) is estimated to improve parameter estimation as well as prediction instead of being considered as a separate parameter(s) of interest.
They used $L_1$-penalties to encourage sparsity in estimates of the regression coefficients and error precision matrix, which can lead to easy-to-interpret
fitted models.  They showed that using the Gaussian loglikelihood, which accounts for the association between the response components, leads to better parameter estimation than using a multivariate residual sum of squares criterion, which does not 
account for this association.  

Similar approaches that jointly estimate the error precision matrix and the regression coefficient matrix have been proposed.
For instance, \citet{lee2012simultaneous} assumed that the response variables and the predictors have a joint multivariate normal distribution, and \citet{Wang2015joint} decomposed the multivariate regression problem into a series of penalized conditional log-likelihood of each response conditional on the predictors and other responses.
\citet{Navon.Rosset.2020} extended this to multivariate linear mixed effects models, and \citet{chang2023robust} considered an extension of \citet{MRCE}'s method to multivariate robust linear regression.  
\citet{zhou2017sparse} and \citet{chan2025drfarm} assumed that the covariance structure of the response variables follows a factor model with latent factors.
\citet{zhu2020convex} considered a convex reparametrization of \citet{MRCE}'s joint optimization problem.  These methods assumed
that the error precision matrix is sparse, which may be unreasonable in some applications with highly correlated response components.
In addition, in high-dimensional settings, the graphical lasso subproblem \citep{friedman2008sparse} used
in the block-wise coordinate decent algorithm of \citet{MRCE} struggles when the error precision matrix is dense.  These dense cases call for
values of the penalty tuning parameter near zero, which lead to algorithm failures or very long computing times when solving the graphical lasso subproblem.

This motivated us to develop a new multivariate linear regression method for high-dimensional settings that works when the error precision matrix is dense.  Similarly to \citet{MRCE}, we jointly estimate the regression coefficient matrix and the error precision matrix by minimizing the negative Gaussian loglikelihood. 
However, unlike their approach, we assume the error covariance matrix has compound symmetry (or equicorrelation).   Our primary goal is to estimate the regression coefficient matrix when the responses are highly 
correlated.  Although our assumed error correlation structure is simple, our estimation 
of the regression coefficient matrix is robust
not only to a moderate degree of misspecification in the error correlation but also to non-sparse structures in the true regression coefficients when the error covariance is correctly specified.
We propose an efficient computational algorithm to compute our estimators
and we also propose an approximation to the
validation likelihood for tuning parameter selection.

\section{Problem setup and proposed estimator}\label{sec:setup}
Let $x_i=(x_{i1},\ldots,x_{ip})^T$ be the $p$-dimensional vector of nonrandom predictor values; let $y_i=(y_{i1},\ldots,y_{iq})^T$ be the observed $q$-dimensional response; and let $\epsilon_i=(\epsilon_{i1},\ldots,\epsilon_{iq})^T$ be the error for the $i$th subject $(i=1,\ldots, n)$.
The multivariate linear regression model assumes that
$y_i$ is a realization of
\begin{align}\label{model.setup}
Y_i = B_*^Tx_i + \epsilon_i,\qquad \quad i=1,\ldots,n,    
\end{align}
where $B_*\in\mathbb{R}^{p\times q}$ is
the unknown regression coefficient matrix; 
and $\epsilon_1,\ldots,\epsilon_n$ are iid $N_q(0,\Sigma_*)$.  
To allow for highly correlated response variables, 
we assume that
\begin{equation} \label{mainassumption}
\Sigma_* = \eta^2_{*} \left\{ (1-\theta_*)I_q + \theta_* 1_q 1_{q}^T \right\},
\end{equation}
where $\eta_*\in(0,\infty)$ and $\theta_*\in [0,1)$ are unknown. 
This implies that ${\rm var}(\epsilon_{i,j}) = \eta^2_{*}$ for $(i,j)\in\{1,\ldots, n\}\times\{1,\ldots, q\}$ and that ${\rm corr}(\epsilon_{i,j}, \epsilon_{i,k}) = \theta_*$
when $j\neq k$.
This implicitly assumes that all
of the response components are on the same scale. 
When they are not, we propose a different method
described in Section \ref{sec:gen}.
Our numerical experiments suggest 
that our estimation of $B_*$, which is our primary target, is robust to 
misspecification of $\Sigma_*$.

To write the negative loglikelihood, we first 
express \eqref{model.setup} in terms of matrices:
let $Y\in\mathbb{R}^{n\times q}$ have $i$th row  $Y_i^T$; let $X\in\mathbb{R}^{n\times p}$ have $i$th row $x_i^T$; and let $E\in\mathbb{R}^{n\times q}$ have $i$th row $\epsilon_i^T$.
Then \eqref{model.setup} is $Y=XB_*+E$.
The negative log-likelihood function evaluated at
$(B, \Omega)$, where $\Omega$ is the variable corresponding to the inverse of $\Sigma_{*}$, is
$$
\mathcal{L}(B,\Omega)={\rm tr}\left[\frac{1}{n}(Y-XB)^T(Y-XB)\Omega\right]-\log |\Omega|.
$$
We propose the following penalized likelihood estimator of $(B_*,\eta_*^2,\theta_*)$:
\begin{equation}\label{eq:penalized.LL}
(\hat B, \hat \eta^2, \hat \theta)=
\argmin_{(B,\eta^2,\theta)\in\mathbb{R}^{p\times q} \times (0,\infty) \times [0,1)} \left\{
\mathcal{L}\left(B, \Omega(\eta^2, \theta)\right)
+ \lambda \sum_{j=1}^p \sum_{k=1}^q |B_{jk}|
\right\},
\end{equation}
where $\Omega(\eta^2, \theta)=\left[\eta^2\left\{\theta 1_q 1_{q}' + (1-\theta) I_q \right\}\right]^{-1}$
and $\lambda\in[0, \infty)$ is a tuning parameter.
Similarly to \citet{MRCE}, the $L_1$ penalty on $B$ encourages sparse regression coefficient estimation.
We label the solution to this optimization as the Multivariate Regression with Compound Symmetry [MRCS] estimator.  
When the response components are on different scales, 
we define an alternative estimator described in
Section \ref{sec:gen}.
The MRCE method \citep{MRCE} replaces $\Omega$ in equation \eqref{eq:penalized.LL} with a general error precision matrix consisting of the $q(q+1)/2$ free parameters that are not parametrized by $\eta$ and $\theta$. Added to the $L_1$-penalty on the regression coefficient, it adopts an additional $L_1$-penalty on the off-diagonal entries of $\Omega$ to encourage sparsity in the precision matrix. Additional simulation results demonstrating the computational failure of MRCE under dense error covariance structures (e.g., compound symmetry) are provided in Section 2 of the Supplementary material.

Our goal is to estimate $B_*$ and predict future response values:  we expect that the assumed error covariance structure in \eqref{mainassumption} will be an oversimplification in many applications.
When $\lambda=0$, it is known that
the regression coefficient estimator
that solves \eqref{eq:penalized.LL} would
be unchanged if $\Omega(\eta^2, \theta)$
were replaced by any positive definite matrix
\citep{MRCE}.  So the assumed error correlation
structure only influences how 
the entries in the regression coefficient matrix estimator are shrunk towards zero when $\lambda >0$.
Our numerical experiments suggest that 
our estimation of $B_*$ is robust to misspecification of the error covariance.

\begin{remark}
An anonymous referee pointed out that our compound symmetry–aware estimator reduces to a special case of methods proposed by \citet{zhou2017sparse} and \citet{chan2025drfarm}, which assume that the covariance structure of the response variables follows a factor model, when the model has a single latent factor with a factor loading vector of ones.
In Section 4 of the supplementary material, we show that, when the compound symmetry structure is correctly specified, our method outperforms the factor model with the oracle number of latent factors while requiring lower computational cost. 
Moreover, because the number of latent factors is typically unknown in practice, our computational advantage becomes increasingly significant as the number of latent factors increases.
Thus, our model can be seen as a special case of factor model, which uses more parameters to model the error structure more flexibly. However, as we discuss in Section 4 of the supplementary material, this added complexity combined with the nonconvex nature of the optimization often cause numerical difficulties in practice.
\end{remark}

\section{Computational algorithms}\label{sec:algorithms}
\subsection{Algorithm for exact computation}
By the derivation illustrated in \ref{sec:appd.derivation}, the penalized estimator in \eqref{eq:penalized.LL} is equivalent to the following:
\begin{equation}\label{main.objective.0}
\begin{split}
\argmin_{(B, \eta^2, \theta)\in\mathbb{R}^{p\times q} \times (0,\infty) \times [0,1)} F_\lambda(B,\eta^2,\theta;Y,X),
\end{split}
\end{equation}
where
\begin{equation*}
\begin{split}
F_\lambda(B,\eta^2,\theta;Y,X)&=\frac{1}{n\eta^2(1-\theta)}\|Y-XB\|_{F}^2 
 - \frac{\theta}{n\eta^2(1-\theta)(1-\theta+q\theta)} 
\|(Y-XB)1_q\|^2 \\
 &\indent+(q-1)\log(1-\theta)+\log(1+\{q-1\}\theta)+q\log(\eta^2)\\
 &\indent+ \lambda \sum_{j=1}^p \sum_{k=1}^q |B_{jk}|.
\end{split}
\end{equation*}
We solve \eqref{main.objective.0} 
by blockwise coordinate descent.
We treat $(\eta^2,\theta)$ as the first block, and $B$ as the second block.
We let the superscript $(k)$ denote the $k$th iterate of each component that is being updated.

For a fixed $\hat B^{(k)}$, we reparametrize the problem via $\alpha=\eta^2(1-\theta),\, \gamma=\eta^2(1+(q-1)\theta)$, which is a one-to-one mapping from $(\eta^2,\theta)$. 
Then \eqref{main.objective.0} with $\hat B^{(k)}$ fixed becomes
\begin{equation}\label{transformed.true.obj}
(\hat\alpha^{(k+1)},\hat\gamma^{(k+1)})=\argmin_{0<\alpha\le \gamma <\infty} 
\left\{
\frac{M_1^{(k)}}{\alpha}-\frac{1}{q}\left(\frac{1}{\alpha}-\frac{1}{\gamma}\right)M_2^{(k)}+\log(\gamma)+(q-1)\log(\alpha)
\right\},  
\end{equation}
where $M_1^{(k)}=\frac{1}{n}\|Y-X\hat B^{(k)}\|_F^2$ and $M_2^{(k)}=\frac{1}{n}\|(Y-X\hat B^{(k)})1_q\|^2$.
Since the objective function in \eqref{transformed.true.obj} is separable, by first order condition, we get 
\begin{align}
&\hat\alpha^{(k+1)}=\frac{qM_1^{(k)}-M_2^{(k)}}{q(q-1)},\label{eq:alpha.form}\\
&\hat\gamma^{(k+1)}=\hat\gamma^{(k+1)}(\hat\alpha^{(k+1)})=\max\left\{\hat\alpha^{(k+1)},\frac{M_2^{(k)}}{q}\right\}\label{eq:gamma.form},
\end{align}
where $M_1^{(k)}\ge M_2^{(k)}/q$ holds by the Cauchy-Schwarz inequality.
Finally,
\begin{align}\label{eta.theta.solver}
((\hat\eta^2)^{(k)},\hat\theta^{(k)})=\left(\hat\alpha^{(k)}+\frac{\hat\gamma^{(k)}-\hat\alpha^{(k)}}{q},\frac{\hat\gamma^{(k)}-\hat\alpha^{(k)}}{\hat\gamma^{(k)}+(q-1)\hat\alpha^{(k)}}\right).
\end{align}
The reparameterization converts the original non-convex optimization problem in $(\eta,\theta)$ into a separable problem in $(\alpha,\gamma)$, where each subproblem becomes a strictly convex scalar optimization admitting closed-form updates.
Thus it improves numerical stability and accelerates convergence while automatically obeying the positivity and ordering constraints. 
Just as importantly, this alternative ($\alpha, \gamma$) formulation also facilitates a more streamlined derivation of the large-sample limit distribution of ($\hat\eta, \hat\theta$) in Theorem 1.

Now, when $\hat\Omega^{(k+1)}((\hat\eta^2)^{(k+1)},\hat\theta^{(k+1)})$ is fixed, the corresponding optimization problem is given by
\begin{align}\label{beta.update}
\hat B^{(k+1)}(\hat\Omega^{(k+1)})=\argmin_B \left\{{\rm tr}\left[\frac{1}{n}(Y-XB)^T(Y-XB)\hat\Omega^{(k+1)}\right]\right\} +\lambda \sum_{j=1}^p \sum_{k=1}^q |B_{jk}|.
\end{align}
We solve \eqref{beta.update} through a function \texttt{rblasso} in the R package \texttt{MRCE} \citep{mrcepackage}.

To improve the performance with a better initializer, we suggest the following procedure.
We first perform $q$ lasso regressions for each response with the same optimal tuning parameter $\hat\lambda_0$ selected with a cross validation.
Let $\hat B_{\hat\lambda_0}$ be the solution.
We refer this as combined lasso initializer.
We initialize the algorithm from $\hat B^{(0)}=\hat B_{\hat\lambda_0}$. 
We set the convergence tolerance parameter $\varepsilon=10^{-7}$.
We summarize the algorithmic procedure in Algorithm \ref{MRCS} below.
Steps 1 and 2 both guarantee a decrease in the objective function value.

\begin{algorithm}
 \caption{Multivariate Regression with Compound Symmetry [MRCS]}
\begin{algorithmic}\label{MRCS}
  \scriptsize
  \STATE For each fixed value of $\lambda$, initialize $\hat B^{(0)}=\hat B_{\hat\lambda_0}$. Set $k=0$ and $((\hat\eta^2)^{(0)},\hat\theta^{(0)})=(1,0)$.
  \STATE {\bf Step 1:} Compute $((\hat\eta^2)^{(k+1)},\hat\theta^{(k+1)})(\hat B^{(k)})$ by \eqref{eta.theta.solver}.
  \STATE {\bf Step 2:} Compute $\hat B^{(k+1)}((\hat\eta^2)^{(k+1)},\hat\theta^{(k+1)})$ by \eqref{beta.update}.
  \STATE {\bf Step 3:} If $|F_\lambda(\hat B^{(k+1)},(\hat\eta^2)^{(k+1)},\hat\theta^{(k+1)})-F_\lambda(\hat B^{(k)},(\hat\eta^2)^{(k)},\hat\theta^{(k)})|<\varepsilon \tr(Y'Y)/n$ then stop. Otherwise go to Step 1 and $k\gets k+1$.
\end{algorithmic}   
\end{algorithm}

\subsection{Algorithm for approximate computation}\label{sec:appx.algo.1}
When $p\ge n$, since the residual sample covariance matrix, $(Y-XB)^T(Y-XB)/n$, is singular or near-singular, the canonical MRCS algorithm's alternating updates can encounter difficulties.
For example, when there exists $\Bar{B}$ such that any one of the columns of $Y-X\Bar{B}$ is zero, then the objective function value for updating $\Omega$ can be made arbitrarily negative by increasing the corresponding diagonal element of $\Omega$ sufficiently, which implies that a global minimizer does not exist. 
In addition, the computational time of the blockwise coordinate descent for fitting both regression coefficients and the error precision matrix tends to be significantly higher in increased dimensions of $p,\,q$ \citep{MRCE}.

To improve computational efficiency and ensure the well-posed optimization problem in high dimensions, we provide an approximate solution to our procedure by following the idea of approximate MRCE [ap.MRCE] proposed in \citet{MRCE}.
ap.MRCE addresses these issues by first obtaining an initial sparse estimate of $\Omega$ and then performing a single update of $B$, with simulations demonstrating its competitive performance.
Similarly, our approximate algorithm limits the number of updates to reduce instability and improve computational speed, while preserving accurate estimation of $B$.

As in Algorithm 1, we compute the initial $\hat B^{(0)}=\hat B_{\hat\lambda_0}$.
Then, for each $\lambda$, we compute $(\hat\eta^2(\hat B_{\hat\lambda_0}),\hat\theta(\hat B_{\hat\lambda_0}))$ by \eqref{eta.theta.solver}.
After this step for (inverse) covariance update, we compute the proposed approximate solution $\hat B$ with known $(\hat\eta^2(\hat B_{\hat\lambda_0}),\hat\theta(\hat B_{\hat\lambda_0}))$ by \eqref{beta.update}.
We suppress the dependence of $\hat B$ and $(\hat\eta^2(\hat B_{\hat\lambda_0}),\hat\theta(\hat B_{\hat\lambda_0}))$ on $\lambda$ for notational simplicity.
We refer the approximate solution to \eqref{main.objective.0} as ap.MRCS.
We summarize the procedure in Algorithm \ref{apMRCS} below.

\begin{algorithm}
 \caption{Approximated version of Multivariate Regression with Compound Symmetry [ap.MRCS]}
\begin{algorithmic}\label{apMRCS}
  \scriptsize
  \STATE For each fixed value of $\lambda$, initialize $\hat B^{(0)}=\hat B_{\hat\lambda_0}$.
  \STATE {\bf Step 1:} Compute $(\hat\eta^2,\hat\theta)(\hat B^{(0)})$ by \eqref{eta.theta.solver}.
  \STATE {\bf Step 2:} Compute $\hat B(\hat\eta^2,\hat\theta)$ by \eqref{beta.update}.
\end{algorithmic}   
\end{algorithm}

\section{Extension to general equicorrelation covariance matrix}\label{sec:gen}

When the response components are on different scales, then it may be unreasonable to assume
that the components of the errors have the same variance $\eta_{*}^2$.  We extend our model by allowing different $\eta_{*j}$'s for each response component. 
The model follows \eqref{model.setup} except that we use the following error covariance matrix:
\begin{align*}
\Sigma_*&= {\rm diag}(\{\eta_{*j}\}_{j=1}^q) \left[(1-\theta)I_q+\theta 1_q1_q^T\right] {\rm diag}(\{\eta_{*j}\}_{j=1}^q),
\end{align*}
where 
$(\eta_{*1},\ldots, \eta_{*q}) \in (0,\infty)^q$
are the unknown standard deviations of the $q$ marginal distributions of the error, i.e. 
${\rm var}(\epsilon_{i,j})=\eta_{*j}^2$ for $(i,j)\in\{1,\ldots, n\}\times\{1,\ldots,q\}$.  
The inverse error covariance matrix is
\begin{align}\label{general.precision}
\Omega_*=\frac{1}{1-\theta_*}{\rm diag}(\{\eta_{*i}^{-1}\}_{i=1}^q)\left[I_q-\frac{\theta_*}{1+(q-1)\theta_*}1_q1_q^T\right]{\rm diag}(\{\eta_{*i}^{-1}\}_{i=1}^q).    
\end{align}
Then the corresponding penalized negative log-likelihood optimization is 
\begin{equation}\label{main.objective.gen}
\begin{split}
\argmin_{(B,\{\eta_i\}_{i=1}^q, \theta)\in\mathbb{R}^{p\times q} \times \mathbb{R}_{+}^q \times [0,1)} F^{gen}_\lambda(B,\{\eta_i\}_{i=1}^q, \theta;Y,X),
\end{split}
\end{equation}
where
\begin{equation*}
\begin{split}
F^{gen}_\lambda(B,\{\eta_i\}_{i=1}^q, \theta;Y,X)&=\frac{1}{n(1-\theta)}\|(Y-XB){\rm diag}(\{\eta_i^{-1}\})\|_F^2\\
 &\indent-\frac{\theta}{n(1-\theta)(1-\theta+q\theta)}\|[(Y-XB){\rm diag}(\{\eta_i^{-1}\})] 1_q\|^2\\
 &\indent+(q-1)\log(1-\theta)+\log(1+\{q-1\}\theta)+2\sum_{i=1}^q\log(\eta_i) \\
&\indent+\lambda \sum_{j=1}^p \sum_{k=1}^q |B_{jk}|,
\end{split}
\end{equation*}
and $\lambda \in [0,\infty)$ is a tuning parameter.
As in Algorithm \ref{sec:appx.algo.1}, we treat $(\{\eta_i\}_{i=1}^q,\theta)$ as the first block and $B$ as the second block in blockwise coordinate descent.

For a fixed $\hat B^{(k)}$, we solve for $(\{\eta_i\}_{i=1}^q,\theta)$ cyclically.
By simple algebra (see \ref{sec:appd.derivation}), we compute the update for $\eta_j$ by
\begin{align}\label{rhos.update}
\hat\eta_j^{(k+1)}=\frac{-K_1^{(k)}+\sqrt{(K_1^{(k)})^2+4K^{(k)}_2}}{2},
\end{align}
where
\begin{align*}
&K_1^{(k)}= \frac{\hat\theta^{(k)}}{n(1-\hat\theta^{(k)})(1+(q-1)\hat\theta^{(k)})}\sum_{i=1}^n e_{ij}^{(k)}\left(\sum_{k\neq j}e_{ik}^{(k)}\frac{1}{\hat\eta^{(k)}_k}\right),\\
&K_2^{(k)}= \frac{1+(q-2)\hat\theta^{(k)}}{n(1+(q-1)\hat\theta^{(k)})(1-\hat\theta^{(k)})}\left(\sum_{i=1}^n (e_{ij}^{(k)})^2\right),
\end{align*}
and $e_{ij}^{(k)}$ is the $(i,j)$-th element of $Y-X\hat B^{(k)}$.
After sequentially solving for $\{\eta_i\}_{i=1}^q$, we solve the following by line search:
\begin{equation}\label{theta.opt.problem}
\begin{split}
\hat\theta^{(k+1)}=\argmin_{\theta\in [0,1)} &\frac{1}{n(1-\theta)}\|(Y-X\hat B^{(k)}){\rm diag}(\{(\hat\eta_i^{(k+1)})^{-1}\})\|_F^2\\
&\indent -\frac{\theta}{n(1-\theta)(1-\theta+q\theta)}\|[(Y-X\hat B^{(k)}){\rm diag}(\{(\hat\eta_i^{(k+1)})^{-1}\})] 1_q\|^2\\
 &\indent+(q-1)\log(1-\theta)+\log(1+\{q-1\}\theta).
\end{split}
\end{equation}
For efficiency, we do not require full convergence in each inner loop over $(\{\eta_i\}_{i=1}^q,\theta)$. Instead, we conduct single-iteration updates for each $\{\eta_i\}_{i=1}^q$ and $\theta$ by \eqref{rhos.update} and \eqref{theta.opt.problem}. 

We compute $\hat B^{(k+1)}$ using 
\eqref{beta.update} with
$\hat\Omega^{(k+1)}$ set to the right side of
\eqref{general.precision} with 
$\{\eta_i\}_{i=1}^q$ and $\theta$ 
replaced by $\{\hat\eta_i^{(k+1)}\}_{i=1}^q$ and $\hat\theta^{(k+1)}$, respectively.
We call the solution to \eqref{main.objective.gen} MRGCS.

We again start the algorithm with combined lasso initializer as in Algorithm \ref{MRCS}.
We compared the combined lasso initializer with
a separate lasso initializer that fits $q$ separate lasso regressions for each response with the different tuning parameters selected by cross validation for each response. 
However, we found that the combined lasso initializer generally performed better than the separate lasso except for a few cases.
%For example, in our simulations, $(s_1,s_2)=(0.5,0.1)$ (the sparsity level in $B_*$) is the only case where the separate lasso initializer performed better than the combined lasso initializer.
%Except for this case, the combined lasso initializer substantially dominated the other consistently.
We again set the convergence tolerance parameter $\varepsilon=10^{-7}$ for the following Algorithm \ref{MRGCS} which summarizes the entire procedure.

\begin{algorithm}
 \caption{Multivariate Regression with Generalized Compound Symmetry [MRGCS]}
\begin{algorithmic}\label{MRGCS}
  \scriptsize
  \STATE For each fixed value of $\lambda$, initialize $\hat B^{(0)}=\hat B_{\hat\lambda_0}$. Set $k=0$ and $(\{\hat\eta_i^{(0)}\}_{i=1}^q,\hat\theta^{(0)})=(1_q,0)$.
  \STATE {\bf Step 1:} Compute one-step update of $(\{\hat\eta_i^{(k+1)}\}_{i=1}^q,\hat\theta^{(k+1)})(\hat B^{(k)})$ by solving \eqref{rhos.update}--\eqref{theta.opt.problem} cyclically.
  \STATE {\bf Step 2:} Compute $\hat B^{(k+1)}(\{\hat\eta_i^{(k+1)}\}_{i=1}^q,\hat\theta^{(k+1)})$ by solving \eqref{beta.update}.
  \STATE {\bf Step 3:} If $|F^{gen}_\lambda(\hat B^{(k+1)},\{\hat\eta_i^{(k+1)}\}_{i=1}^q,\hat\theta^{(k+1)})-F^{gen}_\lambda(\hat B^{(k)},\{\hat\eta_i^{(k)}\}_{i=1}^q,\hat\theta^{(k)})|<\varepsilon \tr(Y'Y)/n$ then stop. Otherwise go to Step 1 with $k\gets k+1$.
\end{algorithmic}   
\end{algorithm}

\subsection{Approximate solutions II}\label{sec:appx.algo.2}
As in Section \ref{sec:appx.algo.1}, we propose an approximate solution to \eqref{main.objective.gen}.
Again, we start the algorithm with $\hat B^{(0)}=\hat B_{\hat\lambda_0}$, the combined lasso initializer. 
Then, for each $\lambda$, we compute $(\{\hat\eta_i\}_{i=1}^q,\hat\theta)$ from \eqref{main.objective.gen}.
However, the difference between this approximation and the canonical MRGCS is that we conduct the updating step for $(\{\hat\eta_i\}_{i=1}^q,\hat\theta)$ until convergence instead of a single-iteration update.
This subproblem is obviously convex in $(\{\eta_i\}_{i=1}^q,\theta)$.
Next, we compute the final $\hat B$ by solving \eqref{beta.update} once. 
We call this approximation ap.MRGCS.
We summarize it in Algorithm \ref{ap.MRGCS} below.

We conducted additional simulations comparing MRCS, MRGCS, and their approximate versions under two initializers, combined and separate lasso, using a much finer grid of tuning parameters (see Section 1.3 of the supplementary material). 
Results show that initialization has little impact on MRCS and ap.MRCS, whereas MRGCS and ap.MRGCS are sensitive under high sparsity, where separate lasso performs better. 
Based on the favorable performance in our simulation settings and the computational simplicity, we adopted the combined lasso initializer. 
However, in practice, we recommend practitioners to evaluate both initialization strategies and choose the one that achieves the best performance via cross-validation.

\begin{algorithm}
 \caption{Approximated version of Multivariate Regression with generalized Compound Symmetry [ap.MRGCS]}
\begin{algorithmic}\label{ap.MRGCS}
  \scriptsize
  \STATE For each fixed value of $\lambda$, initialize $\hat B^{(0)}=\hat B_{\hat\lambda_0}$.
  \STATE {\bf Step 1:} Compute $(\{\hat\eta_i\}_{i=1}^q,\hat\theta)(\hat B^{(0)})$ cyclically by solving \eqref{rhos.update} and \eqref{theta.opt.problem} until convergence.
  \STATE {\bf Step 2:} Compute $\hat B(\{\hat\eta_i\}_{i=1}^q,\hat\theta)$ by solving \eqref{beta.update}.
\end{algorithmic}   
\end{algorithm}

\section{Asymptotic statistical properties}\label{sec:classic.theory}
We study large sample asymptotic behavior of our estimator MRCS defined by \eqref{main.objective.0}
using the same asymptotic framework as \citet{lee2012simultaneous}.
We keep $p$ and $q$ fixed with $n>p+q$ throughout this section.
As stated in Proposition 1 of \citet{zou2006adaptive}, the equally-weighted lasso-penalized estimator does not possess the oracle property.  Thus, to ensure the oracle property, it is encouraged to assign weights, $w_{jk}$, to each $|B_{jk}|$ in the penalty term. 
We let $w_{jk}=1/|B^{({\rm ols})}_{jk}|^r$ ($r>1$) be the weight for $|B_{jk}|$, where $B^{({\rm ols})}_{jk}$ is $(j,k)$-th element of the regression coefficient matrix obtained from ordinary least squares.
Using these weights, the slightly modified optimization (than \eqref{main.objective.0}) we consider in this section is as follows:
\begin{equation}\label{main.objective.weight}
\begin{split}
\argmin_{(B, \eta^2, \theta)\in\mathbb{R}^{p\times q} \times (0,\infty) \times [0,1)}
 &\frac{1}{\eta^2(1-\theta)}\tr\left\{n^{-1}(Y-XB)^T(Y-XB)\left[I_q-\frac{\theta}{1-\theta+q\theta}1_q1^T_q\right] \right\}\\
 &\quad+(q-1)\log(1-\theta)+\log(1+\{q-1\}\theta)+q\log(\eta^2)\\
&\quad+\lambda \sum_{j=1}^p \sum_{k=1}^q w_{jk}|B_{jk}|.
\end{split}
\end{equation}
We make the following assumptions to state large sample asymptotics of MRCS.
\begin{assumptionp}{A}\label{assm:A}\phantom{blah}
\begin{enumerate}
\myitem{(A1)}\label{assm:A1}: $X^TX/n\rightarrow Z$, where $Z$ is a positive definite matrix. 
\myitem{(A2)}\label{assm:A2}: There exists $\Tilde{B}_{jk}$ a $\sqrt{n}$-consistent estimator of $B_{jk*}$, where $B_{jk*}$ is the $(j,k)$-th element of $B_*$ for $(j,k)\in\{1,\ldots,p\}\times\{1,\ldots,q\}$. 
\myitem{(A3)}\label{assm:A3}: There exists $(\Tilde{\eta}^2,\Tilde{\theta})$ $\sqrt{n}$-consistent estimator of $(\eta_*^2,\theta_*)$.
\myitem{(A4)}\label{assm:A4}: The distribution of $E$ has finite joint fourth moments.
\end{enumerate}
\end{assumptionp}
Condition \ref{assm:A1} is a standard condition in linear regression large sample asymptotics literature, and was also assumed in \citet{zou2006adaptive}, \citet{lee2012simultaneous}, and \citet{chang2023robust}.
It implies that the design matrix has good asymptotic behavior.
We note that \ref{assm:A2} and \ref{assm:A3} are generally satisfied by maximum likelihood estimators [MLEs] or $L_2$-penalized MLEs.
\ref{assm:A4} is satisfied by broad class of error distributions such as multivariate sub-exponential distributions.

For a matrix $R\in\mathbb{R}^{p\times q}$, we let ${\rm vec}(R)\in\mathbb{R}^{pq}$ be the vector formed by stacking the columns of $R$.
Define $S:=\{i:{\rm vec}(B_*)_i\neq 0\}$.
Let $v_A$ be the subvector of the
of the entries in $v$ with indices in $A$.
For a square matrix $M\in\mathbb{R}^{q\times q}$, we further define $M_A$ as the $|A|\times |A|$ matrix obtained by removing the $i$-th row and column of M for $i\in A^c$, where $A \subseteq \{1,\ldots,q\}$.
We let $\otimes$ denote the Kronecker product of two matrices.
Recall that the inverse of compound symmetry error covariance is computed as
\begin{align*}
\Omega_*=\frac{1}{\eta^2_*(1-\theta_*)}\left[I_q-\frac{\theta_*}{1+(q-1)\theta_*}1_q1_q^T\right].
\end{align*}

In the following theorem, we provide the oracle guarantee, the limit distribution of $\hat B$, and the joint limit distribution of $(\hat \eta^2,\hat\theta)$.
This result is analogous to Theorem 3 of \citet{lee2012simultaneous}.
\begin{theorem}\label{thm:classic.anal}
Under the conditions \ref{assm:A1}--\ref{assm:A4}, assuming that $\lambda\sqrt{n}\rightarrow0$ and $\lambda n^{(r+1)/2}\rightarrow \infty$, then, there exists a local minimizer $(\hat B,\hat\eta^2,\hat\theta)$ to the optimization problem \eqref{main.objective.weight} that satisfies
\begin{align*}
&\lim_{n\rightarrow \infty} P(\hat B_{jk}=0)=1\qquad \qquad {\rm if~}B_{jk*}=0,\\
&\sqrt{n}({\rm vec}(\hat B)_{S}-{\rm vec}(B_*)_{S})\rightarrow_d N(0,D^{-1}),\\
&\left\|(\hat\eta^2,\hat\theta)-(\eta_*^2,\theta_*)\right\|=O_p(1/\sqrt{n}),
\end{align*}
where $D=(\Omega_* \otimes Z)_{S}$.
Furthermore, if $\theta_*\in (0,1)$, then
\begin{align*}
\sqrt{n}((\hat\eta^2,\hat\theta)-(\eta_*^2,\theta_*))\rightarrow_d W^TN_2(0,V),
\end{align*}
where $W=(1,-1/\eta_*^2)^T$, and $V$ is defined in \ref{appd:proof.classic} (see \eqref{eq:V.elements}).
\end{theorem}

\section{Tuning parameter selection}\label{sec:tuning.parameter}

We suggest the following 
tuning parameter selection procedure for MRCS, ap.MRCS, MRGCS, and ap.MRGCS:
\begin{align}\label{eq:tuning.selection.ours}
\hat\lambda=\argmin_{\lambda\in\Lambda} \sum_{k=1}^{K}\tr\left\{n^{-1}(Y_k-X_k\hat B_{-k})^T (Y_k-X_k\hat B_{-k})\hat\Omega_{-k} \right\},
\end{align}
where 
\begin{align*}
\hat\Omega_{-k}=\left[\hat\eta_{-k}\left\{\hat\theta_{-k} 1_q 1_{q}^T + (1-\hat\theta_{-k}) I_q \right\}\right]^{-1}.
\end{align*}
In \eqref{eq:tuning.selection.ours}, $\Lambda$ is a candidate set of tuning parameter values, $K$ is the number of folds used for the cross validation, $\hat B_{-k}$ is the estimated regression coefficient matrix from each method based on the observations that excludes the $k$th fold, $Y_k,\,X_k$ are the responses and predictors corresponding to $k$th fold, and $\hat\eta_{-k},\,\hat\theta_{-k}$ are the estimated parameters based on the observations that excludes the $k$th fold.
The loss function in \eqref{eq:tuning.selection.ours} is the negative Gaussian validation loglikelihood
without the log determinant term.  
Including this term was slightly less stable than 
excluding it in our numerical experiments. 
Similar Gaussian validation likelihood minimization is presented in \citet{lee2012simultaneous}, which used this loss criterion for finding an optimal tuning parameter associated graphical lasso applied to the error precision. 

We note that we do not use $\hat\eta_{-k},\,\hat\theta_{-k}$ from the output of MRCS.
Rather we use the one-step estimators used for ap.MRCS that are computed in the Step 1 of Algorithm \ref{sec:appx.algo.1}.
We use the same $\hat\eta_{-k},\,\hat\theta_{-k}$ for both MRCS and ap.MRCS.
This is because the stability of ap.MRCS output turned out to be significantly better than that of the canonical MRCS.
Furthermore,  we use $\hat\eta_{-k},\,\hat\theta_{-k}$ which assumes compound symmetry error covariance even when we fit MRGCS or ap.MRGCS which assume varying marginal error variances.
This is because the output of ap.MRCS was significantly more stable in the tuning parameter selection compared with that of ap.MRGCS as well as the MRGCS output.
We provide further support regarding this tuning parameter selection procedure through extensive simulations in Section 1.1 and 1.2 of the supplementary material \citep{MRCS.supp}.

\section{Simulations}\label{sec:simul}
The data for our simulations are generated 
from the following model:
\begin{align}\label{eq:dat.gen.model}
Y_i=\mu+X_iB_*+\epsilon_i,\quad X_i\sim N(0,\Sigma_X)\in\mathbb{R}^p,\quad \epsilon_i\sim N(0,\Sigma_*)\in\mathbb{R}^q,\quad i=1,\ldots,n
\end{align}
where $\mu=(1,1+4/(q-1),\ldots,5)$, $(\Sigma_X)_{i,j}=(0.7)^{|i-j|}$ for $1\le i,j \le p$, and
\begin{align}\label{eq:dat.gen.error}
\Sigma_*={\rm diag}(\{\eta_k\}_{k=1}^q) \left[(1-\theta)I_q+\theta 1_q1_q^T\right] {\rm diag}(\{\eta_k\}_{k=1}^q).
\end{align}
We use $n=50$ training observations throughout.
To generate $B_*$, we follow the same regression coefficient matrix generating procedure used in Section 3 of \citet{MRCE}. 
We define the operation $*$ as the element-wise matrix product.
The coefficients matrix $B_*$ is generated by
\begin{align}\label{default.genB}
    B_* = W * K * Q,
\end{align}
where $W$ has iid entries from $N(0,1)$; $K$ has entries from iid $Ber(s_1)$ (a Bernoulli distribution which returns $1$ with probability $s_1$); $Q=1_p1_q^TQ_1$ where a $q\times q$ diagonal matrix $Q_1$ has diagonal elements from iid $Ber(s_2)$.
This results in, $Q$ has rows that are either all one or all zero, where the population proportion of having all-one row ($1_q'$) equals to $s_2$.
Under this setting each model is expected to have $(1-s_2)p$ predictors that are irrelevant for all $q$ responses, and each relevant predictor is expected to have contribution to $s_1q$ of the response variables.

We compare a separate lasso regression with a uniquely selected tuning parameter for each response [LASP], combined lasso [LAS] which employs the same tuning parameter for all responses, MRCE, approximate MRCE [ap.MRCE] \citep{MRCE}, MRCS (the solution to \eqref{main.objective.0}), MRGCS (the solution to \eqref{main.objective.gen}).
In addition to the canonical estimators, MRCS and MRGCS, we also consider their approximate versions: ap.MRCS (see Section \ref{sec:appx.algo.1}), as well as ap.MRGCS (see Section \ref{sec:appx.algo.2}).
Moreover, we compare these methods to an oracle procedure that assumes the true error precision $\Omega_*=\Sigma_*^{-1}$ is known, and only estimates $B$ by $L_1$-penalty.
\begin{equation}\label{eq:MRCS-Or}
\begin{split}
\hat B_{\rm Or}:=\argmin_{B\in\mathbb{R}^{p\times q}}
 & \tr\left\{n^{-1}(Y-XB)^T(Y-XB)\Omega_* \right\}+ \lambda \sum_{j=1}^p \sum_{k=1}^q |B_{jk}|.
\end{split}
\end{equation} 
We refer this estimator as MRCS-Or.
In MRCE, we penalized the diagonals of the inverse covariance matrix only when $p\geq n$.

Tuning parameters are selected using 5 fold cross validation from $\lambda\in\Lambda$, where $\Lambda=\{10^{-4+0.5k}:k=0,1,\ldots,14\}$.
The optimal tuning parameter selection procedures for our methods are discussed in Section \ref{sec:tuning.parameter} (see \eqref{eq:tuning.selection.ours}).
The optimal tuning parameter for MRCS-Or is also selected by \eqref{eq:tuning.selection.ours} except that we use $\Omega_*$ in place of $\hat\Omega_{-k}$ in \eqref{eq:tuning.selection.ours}.
For the other competitors, we select each tuning parameter that minimizes validation prediction error with 5 fold cross validation.
For the estimators which require a single tuning parameter, we again use the candidate set as $\lambda\in \{10^{-4+0.5k}:k=0,1,\ldots,14\}$.
And for MRCE, we use $(\lambda_1,\lambda_2)\in\Lambda_1\times \Lambda_2$, where $\Lambda_1=\Lambda_2=\{10^{-4+0.5k}:k=0,1,\ldots,14\}$.

For the primary criterion for the model comparison, we measure the model error, ${\rm tr}[(\hat B-B_*)^T \Sigma_X (\hat B-B_*)]$ \citep{breiman1997predicting,yuan2007model} with $\hat B$ provided by each method.
We also measure the prediction error, $\|\hat Y-Y\|_F^2$, on the test set which has 200 observations which are generated from the same data generating process as in the training set.
We further measure the true negative rate [TNR] and true positive rate [TPR] for the regression coefficient matrix estimation as follows \citep{MRCE}:
\begin{align}\label{TNR}
{\rm TNR}(\hat B)=\frac{\#\{(i,j)\in[p]\times[q]:\hat B_{ij}=0,\,B_{ij*}=0\}}{\#\{(i,j)\in[p]\times[q]:B_{ij*}=0\}},   
\end{align}
\begin{align}\label{TPR}
{\rm TPR}(\hat B)=\frac{\#\{(i,j)\in[p]\times[q]:\hat B_{ij}\neq 0,\,B_{ij*}\neq 0\}}{\#\{(i,j)\in[p]\times[q]:B_{ij*}\neq 0\}}, \end{align}
where $[k]=\{1,\ldots,k\}$ for $k\in\mathbb{N}$; $\# A$ stands for the number of elements in the set $A$.

\subsection{Setting I: Constant $\eta$}\label{sec:fixed.eta.simul}
In this section, we consider $(p,q)\in\{(20,50),(50,20),(80,80)\}$,
We refer the readers to Section 3.1 in the supplementary material \citep{MRCS.supp} for the results of low dimensional simulations ($p=q=20$).
We vary $s_1\in\{0.1,0.5\},\,s_2\in\{0.1,0.5,1\},\theta\in\{0,0.5,0.75,0.9,0.95\}$, and fix $\eta_i=1$ for all $i=1,\ldots,q$. 
We drop MRCE due to its substantially high computation time.
When $p\geq n$ we exclude MRCS, MRGCS from the set of competitors, where we still consider ap.MRCS, and ap.MRGCS, since they can avoid the residual covariance instability (see Section \ref{sec:appx.algo.1}).
The results for each setting are based on 50 independent replications.
Throughout the results, the trend across the estimators in the prediction error is nearly equivalent to that in the model error.
Hence, we mainly discuss the model error as our primary criterion.

\subsubsection{Results when $(p,q)=(20,50)$}\label{sec:highdim.pt1}

Complete simulation results are provided in Figure 26--29 of the supplementary material \citep{MRCS.supp}. 
As $\theta$ increases, all equicorrelation-based estimators consistently outperform ap.MRCE and the two lasso methods. 
Among non-oracle estimators, MRCS and ap.MRCS performed best, with MRGCS and ap.MRGCS also competitive. 
For TNR, sep.lasso achieved the highest rates, followed by MRCS and ap.MRCS, while ap.MRCE performed worst. 
In contrast, for TPR, the lasso methods were generally poorest, and all joint optimization methods achieved higher TPR as $\theta$ increased, indicating their tendency to provide denser solutions.

\subsubsection{Results when $(p,q)=(50,20)$}\label{sec:highdim.pt2}

Complete simulation results are in Figure 30--33 of the supplementary material \citep{MRCS.supp}. 
ap.MRCS was the best non-oracle estimator, with ap.MRGCS performing similarly except for $(s_1,s_2)=(0.5,0.1)$, where ap.MRCE outperformed it when $\theta=0$ and $0.95$. 
Overall, our non-oracle estimators outperformed ap.MRCE and the lasso methods. 
The TNR/TPR patterns were similar to those for $(p,q)=(20,50)$, though ap.MRCE showed improved TNR.

\subsubsection{Results when $(p,q)=(80,80)$}\label{sec:highdim.pt3}

The computing time for ap.MRCE was substantially higher than that of other procedures, particularly when the tuning parameter associated with the graphical lasso subproblem is close to $10^{-4}$. 
To address this, we restricted the tuning parameter set for the graphical lasso penalty into $\{10^{-2+0.5k}:k=0,1,\ldots,8\}$. 
This setup was also used for ap.MRCE in the simulations with $(p,q)=(80,80)$ in Section \ref{sec:varied.eta.asym} and Section 3.2 of the supplementary material \citep{MRCS.supp}. 
Complete results are provided in Figure 34--37 of the supplementary material. Among the non-oracle estimators, ap.MRCS performed best, while ap.MRGCS only underperformed when $(s_1,s_2)=(0.5,0.1)$ and remained competitive otherwise. In contrast, ap.MRCE exhibited significantly poorer performance for this $(p,q)$ setting.

Representative model error comparison plots for the case $(s_1,s_2) = (0.5,0.5)$ under the above three $(p,q)$ settings are illustrated in Figure \ref{fig:section7-fig1}--\ref{ffig:section7-fig3}.

\subsection{Setting II: Equicorrelation of $\Sigma_*$ with heterogeneous and asymmetrically distributed $\eta_i$'s}\label{sec:varied.eta.asym}
In this section, we study the performance of our equicorrelation-based estimators under generalized $\Sigma_*$ in which $\eta_i$'s are now heterogeneous.
In Section 3.2 of the supplementary material \citep{MRCS.supp}, we compared our methods to the others in the settings where $\eta_i$'s are still heterogeneous but symmetrically distributed.
We discovered that MRCS and ap.MRCS showed comparable performance to MRGCS and ap.MRGCS, the two best non-oracle methods in that setting.  
This may be due to the symmetrically distributed $\eta_i$'s ($i\in\{1,\ldots,q\}$) which can smooth out heterogeneous marginal error variances. 
Now we consider asymmetric cases in this section. 
We used the same data generating process \eqref{eq:dat.gen.model}--\eqref{default.genB} except that we considered
\begin{equation}\label{eq:varied.etas.asym.20.50}
\begin{split}
&\eta_1,\ldots,\eta_{10}=1/2,\, \eta_{11},\ldots,\eta_{20}=1/\sqrt{2},\,\eta_{21},\ldots,\eta_{30}=1,\\
&\eta_{31},\ldots,\eta_{40}=\sqrt{3},\, \eta_{41},\ldots,\eta_{50}=3, 
\end{split}
\end{equation}
when $(p,q)=(20,50)$, and 
\begin{equation}\label{eq:varied.etas.asym.50.20}
\begin{split}
&\eta_1,\ldots,\eta_{4}=1/2,\, \eta_{5},\ldots,\eta_{8}=1/\sqrt{2},\,\eta_{9},\ldots,\eta_{12}=1,\\
&\eta_{13},\ldots,\eta_{16}=\sqrt{3},\, \eta_{17},\ldots,\eta_{20}=3,
\end{split}
\end{equation}
when $(p,q)=(50,20)$.
Lastly, when $(p,q)=(80,80)$, we used
\begin{equation}\label{eq:varied.etas.asym.80.80}
\begin{split}
&\eta_1,\ldots,\eta_{10}=1/2,\, \eta_{11},\ldots,\eta_{20}=1/\sqrt{2},\,\eta_{21},\ldots,\eta_{30}=1/\sqrt[4]{2},\,\eta_{31},\ldots,\eta_{40}=1\\
&\eta_{41},\ldots,\eta_{50}=\sqrt{3},\, \eta_{51},\ldots,\eta_{65}=2,\,\eta_{66},\ldots,\eta_{80}=3.
\end{split}
\end{equation}

To see the effect of varying $\eta_i$ on the regression coefficient shrinkage, we also compared the original $L_1$-penalty on $B$, $\lambda\sum_{j,k}|B_{jk}|$, with an adaptive $L_1$-penalty, $\lambda\sum_{j,k}|\eta_k^{-1}B_{jk}|$, for MRCS-Or.
However, we did not find any supporting evidence on the use of the adaptive penalty instead of the original penalty. 
Thus, we still suggest the use of the original $L_1$-penalty function even in the case of varying $\eta_i$.

\subsubsection{Results when $(p,q)=(20,50)$}\label{sec:varied.asymm.pt1}

Complete results are shown in Figures 38–41 of the supplementary material \citep{MRCS.supp}. Among the non-oracle methods, MRGCS and ap.MRGCS performed best, with MRGCS slightly outperforming its approximate version. Unlike the symmetric setting detailed in Section 3.2.1 of the supplementary material \citep{MRCS.supp}, MRCS and ap.MRCS performed poorly under the asymmetric setting considered here. For larger $s_1s_2$ values, ap.MRCE even outperformed these two methods at higher $\theta$ values. The TNR/TPR patterns are similar to those in Section \ref{sec:highdim.pt1}.

\subsubsection{Results when $(p,q)=(50,20)$}\label{sec:varied.asymm.pt2}

Complete results are in Figure 42--45 of the supplementary material \citep{MRCS.supp}.
ap.MRGCS was the best non-oracle method, although it suffered again when $(s_1,s_2)=(0.5,0.1)$.
As opposed to the results in $(p,q)=(20,50)$, ap.MRCS was the second-best non-oracle competitor, and it outperformed ap.MRGCS when $s_2=0.1$.
ap.MRCE showed poor performance generally.
Note, the prediction performance of MRCS-Or, ap.MRGCS, and ap.MRCS was substantially better than the others.

\subsubsection{Results when $(p,q)=(80,80)$}\label{sec:varied.asymm.pt3}

Complete simulation results are in Figures 46–49 of the supplementary material \citep{MRCS.supp}. In this setting, ap.MRGCS was the best non-oracle estimator, except when $(s_1,s_2)=(0.5,0.1)$, where it was outperformed by ap.MRCE and sep.lasso (similar to the symmetric case in Section 3.2.3 of the supplementary material \citep{MRCS.supp}). ap.MRCS was also outperformed by ap.MRCE for some $\theta$ values when $(s_1,s_2)=(0.1,0.5)$ or $(0.1,1)$.

Representative model error comparison plots for the case $(s_1,s_2) = (0.5,0.5)$ under the above three $(p,q)$ settings are illustrated in Figure \ref{fig:section7-fig4}--\ref{fig:section7-fig6}.

\section{Simulation under the model misspecification}\label{sec:misspec.simul}
\subsection{Misspecification of $B_*$: When the true $B_*$ is non-sparse}\label{sec:dense.B}

In this section, we study the impact of a non-sparse true regression coefficient matrix on the proposed methods. As in previous simulations, we evaluate performance using prediction and model error under the same data-generating process given in \eqref{eq:dat.gen.model}–\eqref{eq:dat.gen.error}. We consider three settings for $\eta_i$'s in \eqref{eq:dat.gen.error}: constant (Setting A), heterogeneous and symmetric (Setting B), and heterogeneous and asymmetric (Setting C).

Unlike earlier experiments, we do not vary the sparsity of $B_*$. Rather than using \eqref{default.genB}, each element of $B_*$ is generated independently from a uniform distribution on $(-1/4,1/4)$ for $(p,q)=(20,50)$ and $(50,20)$, and on $(-1/10,1/10)$ for $(p,q)=(80,80)$. This design yields a non-sparse $B_*$ with many small nonzero signals, which is expected to challenge sep.lasso and comb.lasso.
We again set $\theta\in{0,0.5,0.75,0.9,0.95}$ and include the same competitors as in Section \ref{sec:simul}, with the addition of combined ridge [c.ridge] and separate ridge [s.ridge], defined analogously to combined and separate lasso.

Complete results are in Figure 50--51 in the supplementary material \citep{MRCS.supp}.
Under Setting A, we label the sub-settings $(p,q)=(20,50)$, $(50,20)$, and $(80,80)$ as A1, A2, and A3, respectively, and use similar labels for Settings B and C. 
Across all settings, MRGCS and ap.MRGCS were the best non-oracle methods, outperforming ap.MRCE as well as lasso and ridge estimators as $\theta$ increased. ap.MRCE ranked as the third-best non-oracle method.
In Setting A, with constant $\eta_i$’s, MRCS and ap.MRCS also performed competitively as the best non-oracle methods. However, unlike the sparse case with symmetric and heterogeneous $\eta_i$’s (see Section 3.2 of the supplementary material \citep{MRCS.supp}), both methods struggled under non-sparse designs (Settings B and C). 
The results from this section indicate that if non-sparse $B_*$ has sufficiently small (in absolute value) elements, then there is more gain by joint optimization than lose by misspecified $L_1$ penalty.

\subsection{Misspecification of $\Sigma_*$: Corrupted compound symmetry to general non-sparse covariance}\label{sec:misspec.omega}

In this section, we conduct simulation studies to study the effect of misspecified $\Sigma_*$ on our methods.
Specifically, we analyze how our equicorrelation-based methods perform when the true error covariance $\Sigma_*$ follows a corrupted compound symmetry structure, defined in \eqref{gen.cov.form} below.
Data are generated using the same procedure as in Section \ref{sec:simul}, except that we consider different sparsity levels for $B_*$ (described below) and replace $\Sigma_*$ with the corrupted compound symmetry structure computed as follows:
\begin{align}\label{gen.cov.form}
\Sigma_*=(1-\omega) \left[0.5\left((1-0.9)I_q+0.91_q1_q^T\right)\right] + \omega VDV^{T},
\end{align}
where the columns of $V\in\mathbb{R}^{q\times q}$, $v_1,\ldots,v_{q}$ are generated by the Gram-Schmidt orthogonalization of a $q\times q$ random matrix whose entries are iid standard normal draws; $D\in \mathbb{R}^{q\times q}$ is a diagonal matrix.
For the diagonal entries of $D$, we generated iid draws from a two point distribution that is explained below. Through this we vary the condition number of $\Sigma_*$. 
\eqref{gen.cov.form} can be interpreted as a convex combination of a general poorly conditioned matrix and a compound symmetry with variance $0.5$ and correlation $0.9$, which is also poorly conditioned.
We considered the corruption level $\omega\in\{0.05,0.5,1\}$.
$\omega=1$ leads to a general non-sparse $\Sigma_*$ that is completely different than compound symmetry.
As $\omega$ decreases, the similarity extent of $\Sigma_*$ to the compound symmetry class increases.
Denoting Ber$(p,a,b)$ by the generic two point distribution that returns a numeric value $a$ (resp. $b$) with the probability $p$ (resp. $1-p$), the following is the condition numbers according to each setting:
\begin{itemize}
    \item When $(p,q)=(20,50)$, \begin{itemize}
        \item $D$ constructed upon iid Ber$(0.5,0.1,10)$ draws: 100 / 189.9075 / 412.7517 (setting A / B / C according to $w=1/0.5/0.05$)
        \item $D$ constructed upon iid Ber$(0.4,0.01,10)$ draws: 1000 / 507.617 / 453.8222 (setting D / E / F according to $w=1/0.5/0.05$)
    \end{itemize}
    \item When $(p,q)=(50,20)$, \begin{itemize}
        \item $D$ constructed upon iid Ber$(0.2,0.2,10)$ draws: 50 / 64.29671 / 153.7895 (setting G / H / I according to $w=1/0.5/0.05$)
        \item $D$ constructed upon iid Ber$(0.6,0.02,10)$ draws: 500 / 265.719 / 186.6711 (setting J / K / L according to $w=1/0.5/0.05$)
    \end{itemize}
\end{itemize}
Note that a $q\times q$ compound symmetry matrix with $q=50$ (resp. $q=20$) and $\theta=0.9$ has condition number 451 (resp. 181).
This setting for $D$ avoids a diagonally dominant $\Sigma_*$, that behaves like a diagonal matrix.
Additional experiments using diagonally dominant $\Sigma_*$ with diagonals of $D$ from a chi-square distribution) are reported in Section 3.3 of the supplementary material \citep{MRCS.supp}.
For the sparsity of $B_*$, we used $(s_1,s_2)\in\{(0.5,0.5),\,(0.5,1),\,(1,1)\}$ where the last pair accounts for dense $B_*$.
Sub-settings A-1, A-2, and A-3 correspond to these sparsity levels, with analogous notation for settings B–L.
The same competitor methods as in Section \ref{sec:dense.B} are considered.

Complete results are in Table 5--8 in the supplementary material \citep{MRCS.supp}.
In general, the oracle estimator performed best, except in cases A-3, D-3, G-3, H-3, I-3, J-3, K-3, L-3, where sparsity is violated and ridge estimators performed the best.
When $B_*$ was sparse and $\omega=1$ or $0.5$, our estimators not only performed comparably to lasso methods but also outperformed other competitors in many scenarios.
For low corruption levels ($\omega=0.05$), our equicorrelation-based estimators outperformed the others substantially.
Even when the true error covariance had varying marginal variances, MRCS and ap.MRCS performed nearly as well as MRGCS and ap.MRGCS and often outperformed the latter two.
The results from this section (and Section 3.3 of the supplementary material \citep{MRCS.supp}) indicate that our methods are robust to model misspecification except when $B_*$ is nearly dense with large signals and the true error covariance deviates substantially from an equicorrelation.

\section{Data analysis}\label{sec:data}
We examined a data set from Korea water resources management information system 
(WAMIS) on the river flow chart of the Han river and its branches. The entire raw data set is public and can be downloaded from the portal (\url{http://www.wamis.go.kr/wkw/flw_dubobsif.do}).
We used the daily stream flow measured in $m^3/s$ throughout the calendar year 2023 from February 1st to December 25th. 
The 6 selected measurement points (bridges) for the analysis are:
Haengju (HJ), Hankang (HK), Gwangjin (GJ), Paldang (PD), Yeojoo (YJ), Nahmhankang (NHK).
Within these 6 points, when considering the average flux in February to April as a reference, flows in YJ, NHK are generally on the same scale (70--100 $m^3/s$), but the other 4 points are generally on much higher scale (120--300 $m^3/s)$.
We set the predictors as the stream flux observations from 1 to 7 days before and the responses as the daily flux at each measurement point. 
Hence, there are 42 predictors and the leading intercept term, as well as 6 response variables in the model.
We used a training set that has 52 observations from February 8th to March 31th.
This indicates $(n,p,q)=(52,43,6)$.
The test set has 139 observations from August 8th to December 25th.
The competitors are MRCE, ap.MRCE, comb.lasso, sep.lasso, MRCS, MRGCS, ap.MRCS, ap.MRGCS, and the factor model [DrFARM; \citet{zhou2017sparse}, \citet{chan2025drfarm}] with the number of latent factors equal to one.
The labels are defined in Section \ref{sec:simul}.

Prediction results are summarized in Table \ref{table:realtable}--\ref{table:realtable.dat2}. 
For the overall average (Table \ref{table:realtable}), ap.MRCS performed best, followed by MRCS, with DrFARM, MRCE, ap.MRCE, and combined lasso next. ap.MRGCS and MRGCS slightly underperformed compared to combined lasso but outperformed separate lasso and OLS. By individual responses (Table \ref{table:realtable.dat2}), MRCS was best in two cases, while ap.MRCS led in three.

\begin{table}[h!]
\centering
\resizebox{\columnwidth}{!}{%
\begin{tabular}{ |c||c|c|c|c|c|c|c|c|c|c| }
 \hline
 \multicolumn{11}{|c|}{Performance comparison table} \\
 \hline
 Competitors & OLS & comb.Lasso & sep.Lasso & MRCS & MRGCS & ap.MRCS & ap.MRGCS & MRCE & ap.MRCE & DrFARM\\
 \hline
River flow & 434634.6 & 265887.3 & 283362.8 & {\bf 246039.1} & 268250.1 & {\bf 241223.4} & 268187.2 & 261599.1 & 263000.8 & 252807.3 \\
 \hline
\end{tabular}
}
\caption{The prediction performance comparison table for the combined response. We measured grand averaged squared error from 139 observations in the test set; $\|\hat Y-Y\|_F^2/139$. Boldface indicates the best model [ap.MRCS] and its canonical version [MRCS] which is the second-best.}
\label{table:realtable}
\end{table}

\begin{table}[h!]
\centering
\resizebox{\columnwidth}{!}{%
\begin{tabular}{ |c||c|c|c|c|c|c| }
 \hline
 \multicolumn{7}{|c|}{Performance comparison table by each response} \\
 \hline
 Competitors $\backslash$ Responses & HJ & HK & GJ & NHK & YJ & PD \\
 \hline
OLS & 1438608  & {\bf 179435.4}  & 217457.6  & {\bf 76625.49}  & 158738.3  & 536943.0 \\ \hline
sep.Lasso & 601845.6  & 246680.3  & {\bf 200497.8}  & 89415.67  & 147132.8  & 309751.8 \\ \hline
comb.Lasso & 601845.6  & 246680.3  & 231999.8  & 95087.78  & 159095.3  & 365468.2 \\ \hline
MRCS & 537848.0  & 239304.6  & {\bf 200497.8}  & 79507.57  & {\bf 134387.3}  & 284689.6 \\ \hline
MRGCS & 533929.2  & 229286.2  & 210172.7  & 81730.28  & 164394.0  & 389988.1  \\ \hline
ap.MRCS & {\bf 515088.8}  & 233800.8  & {\bf 200497.8}  & 79817.72  & 134401.6  & {\bf 283733.4} \\ \hline
ap.MRGCS & 537831.2  & 224128.6  & 214104.0  & 101772.8  & 164287.5  & 366998.9 \\ \hline
MRCE & 525804.0  & 227598.0  & 210888.4  & 80369.34  & 164309.4  & 360625.6 \\ \hline
ap.MRCE & 531651.2  & 223977.3  & 215168.6  & 101128.7  & 163408.5  & 342670.8 \\ \hline
DrFARM & 516587.9  & 230183.7  & 216579.5  & 100864.0  & 162076.3  & 290552.3 \\ \hline
\end{tabular}
}
 \caption{The river flow prediction performance comparison table for each response.  We measured average squared error for each response variable from 139 observations in the test set;  $(\|\hat Y^i-Y^i\|^2)/139$, for $1\le i\le 6$. Boldface indicates the best model for each response.}
\label{table:realtable.dat2}
\end{table}

\section{Conclusion}
We propose a new set of methods for multivariate linear regressions that estimate a dense error covariance matrix in high dimensional settings. As our methods leverage an equi-correlation structure, if modeling assumptions are met, then we expect to see performance gains as the number of responses increases.  This is due to the larger number of responses to estimate the common correlation parameter.

Furthermore, in general, our methods are designed to work in both high-dimensional and low-dimensional settings; we recommend using MRCS and MRGS in low-dimensional settings, and we recommend using ap.MRCS and ap.MRGS in high-dimensional settings as they reduce numerical instability, improve computational efficiency, and give accurate estimation of the regression coefficient matrix.

\section{Acknowledgment}
The authors would like to express gratitude to Professor Aaron Molstad for helpful comments on the proof details. Computational resources were provided by the WVU Research Computing Thorny Flat HPC cluster, which is funded in part by NSF OAC-1726534. Price was partially supported by the National Institute of General Medical Sciences, 5U54GM104942-04, and National Institute of Minority Health and Health Disparities, 5R21MD020187-02. The content is solely the responsibility of the authors and does not necessarily represent the official views of the NIH. 

\appendix

\section{Derivation of the optimization problems \eqref{main.objective.0}}\label{sec:appd.derivation}

We first check the derivation of the optimization problem \eqref{main.objective.0}.
Recall that
\begin{align*}
\Sigma_* = \eta^2_{*} \left\{ (1-\theta_*)I_q + \theta_* 1_q 1_{q}^T \right\}
\end{align*}
Since ${\rm det}(I+uv^T)=1+u^Tv$ for $u,\,v\in\mathbb{R}^q$ and ${\rm det}(cA)=c^q A$ for $A\in\mathbb{R}^{q\times q}$, we have
\begin{align*}
{\rm det}(\Sigma_*) 
& = \eta_*^{2q} (1-\theta_*)^q \left(1 + q \frac{\theta_*}{1-\theta_*}\right).
\end{align*}
This yields 
\begin{align*}
\log {\rm det}(\Sigma_*)  & = q\log (\eta_*^2) + (q-1) \log(1-\theta_*) + \log(1+(q-1)\theta_*).
\end{align*}
To compute $\Sigma_*^{-1}$, we use the following so-called Woodbury identity:
$$
(A + CDC')^{-1} = A^{-1} - A^{-1} C (D^{-1} + C'A^{-1}C)^{-1}C' A^{-1},
$$
with $A = (1-\theta_*) I_q$, $D = \theta_*$, and $C = 1_q$.
Then $A^{-1} = (1-\theta_*)^{-1} I_q$, $D^{-1} = \theta_*^{-1}$, and
\begin{align*}
\Sigma_*^{-1} & =   (\eta_*^2)^{-1}\left( (1-\theta_*)I_q + \theta_* 1_q 1_{q}^T \right)^{-1} \\
 & = (\eta_*^2)^{-1}  (A + CDC^T)^{-1} \\
& =(\eta_*^2)^{-1} \left[A^{-1} - A^{-1} C (D^{-1} + C'A^{-1}C)^{-1}C^T A^{-1}\right]\\
& =(\eta_*^2)^{-1}\left[(1-\theta_*)^{-1} I_q - (1-\theta_*)^{-1} I_q 1_q (\theta_*^{-1} + (1-\theta_*)^{-1}  1_q^T I_q 1_q)^{-1} 1_q^T (1-\theta_*)^{-1} I_q \right]\\
& =(\eta_*^2)^{-1}(1-\theta_*)^{-1} \left( I_q - \frac{\theta_*}{1-\theta_*+q\theta_*} 1_q 1_{q}^T\right).
\end{align*}
This identifies
\begin{equation}
\begin{split}
F_\lambda(B,\eta^2,\theta;Y,X) &= \frac{1}{\eta^2(1-\theta)}\tr\left\{n^{-1}(Y-XB)^T(Y-XB)\left[I_q-\frac{\theta}{1-\theta+q\theta}1_q1^T_q\right] \right\}\\
 &\indent+(q-1)\log(1-\theta)+\log(1+\{q-1\}\theta)+q\log(\eta^2)\\
 &\indent+\lambda \sum_{j=1}^p \sum_{k=1}^q |B_{jk}|.
\end{split}
\end{equation}
Its equivalence to the optimization problem \eqref{main.objective.0} is based on the following:
\begin{equation*}
\begin{split}
&\frac{1}{\eta^2(1-\theta)}\tr\left\{n^{-1}(Y-XB)^T(Y-XB)\left[I_q-\frac{\theta}{1-\theta+q\theta}1_q1^T_q\right] \right\}\\
&=\frac{1}{n\eta^2(1-\theta)}\|Y-XB\|_{F}^2 
 - \frac{\theta}{n\eta^2(1-\theta)(1-\theta+q\theta)} 
\tr\left\{(Y-XB)^T(Y-XB)1_q 1_{q}^T\right\}\\
&=\frac{1}{n\eta^2(1-\theta)}\|Y-XB\|_{F}^2 
 - \frac{\theta}{n\eta^2(1-\theta)(1-\theta+q\theta)} \|(Y-XB)1_q\|^2.
\end{split}
\end{equation*}

We now check the derivation of the optimization problem \eqref{main.objective.gen}.
We cyclically update $\eta_j$ while fixing $B$, $\theta$, and $\eta_{-j}$, where $\eta_{-j}=(\eta_1,\ldots,\eta_{j-1},\eta_{j+1},\ldots,\eta_q)$ for each $j\in\{1,\ldots,q\}$.
By simple algebra, one can derive that the objective function with respect to $\eta_j$ is computed as the following:
\begin{equation*}
\begin{split}
g_j(\eta_j)&=\frac{1}{n(1-\theta)}\sum_{i=1}^n e_{ij}^2\frac{1}{\eta_j^2}-\frac{\theta}{n(1-\theta)(1+(q-1)\theta)}\sum_{i=1}^n\left(\sum_{j=1}^q e_{ij}\frac{1}{\eta_j}\right)^2+2\log(\eta_j),
\end{split}
\end{equation*}
We suppress the dependence of $g_j$ on $\eta_{-j}$, $B$, and $\theta$ for notational simplicity.
From the above computation, we have
\begin{align*}
\frac{1}{2}g_j'(\eta_j)&=\frac{1}{\eta_j}+\frac{\theta}{n(1-\theta)(1+(q-1)\theta)}\sum_{i=1}^n\left(\sum_{k\neq j} e_{ij}e_{ik}\frac{1}{\eta_j^2\eta_k}\right)\\
&\indent -\frac{1+(q-2)\theta}{n(1+(q-1)\theta)(1-\theta)}\left(\sum_{i=1}^n e_{ij}^2\right)\frac{1}{\eta_j^3}.
\end{align*}
Then $g_j'=0$ yields \eqref{rhos.update}, since $\eta_j>0$.

\section{Proof of Theorem \ref{thm:classic.anal}}\label{appd:proof.classic}

We first provide two lemmas which are used for the proof of Theorem \ref{thm:classic.anal}.
Recall $S=\{i:{\rm vec}(B_*)_i\neq 0\}$, and
\begin{align*}
\Omega_*=\frac{1}{\eta^2_*(1-\theta_*)}\left[I_q-\frac{\theta_*}{1+(q-1)\theta_*}1_q1_q^T\right].
\end{align*}

\begin{lemma}\label{lem:consist.Beta}
Under the conditions \ref{assm:A1} and \ref{assm:A3}, assuming that $(\hat\eta^2,\hat\theta)$ are $\sqrt{n}$-consistent estimators of $(\eta_*^2,\theta_*)$, then, if $\lambda\sqrt{n}\rightarrow0$ and $\lambda n^{(r+1)/2}\rightarrow \infty$, the following holds.
\begin{align*}
&\lim_{n\rightarrow \infty} P(\hat B_{jk}=0)=1\qquad \qquad {\rm if~}B_{jk*}=0,\\
&\sqrt{n}({\rm vec}(\hat B)_{S}-{\rm vec}(B_*)_{S})\rightarrow_d N(0,D^{-1}),   
\end{align*}
where $\hat B$ is a solution to the optimization problem \eqref{main.objective.weight} for fixed $(\hat\eta^2,\hat\theta)$, and $D=(\Omega_* \otimes Z)_{S}$.
\end{lemma}
\begin{proof}
The diagonals of $\Omega_*$ are $\frac{1+(q-2)\theta_*}{\eta^2_*(1-\theta_*)(1+(q-1)\theta_*)}=\frac{1}{\eta^2_*(1-\theta_*)}(1-\frac{\theta_*}{1+(q-1)\theta_*})$ and the off-diagonals are $-\frac{\theta_*}{\eta^2_*(1-\theta_*)(1+(q-1)\theta_*)}$.
Let $v_*=\frac{1}{\eta^2_*(1-\theta_*)}$ and $w_*=\frac{\theta_*}{\eta^2_*(1-\theta_*)(1+(q-1)\theta_*)}$.
We further denote $\hat v=\frac{1}{\hat\eta^2(1-\hat\theta)}$ and $\hat w=\frac{\hat\theta}{\hat\eta^2(1-\hat\theta)(1+(q-1)\hat\theta)}$.
Then,
\begin{align*}
|\hat v- v_*|&=\left|\frac{\hat\eta^2(1-\hat\theta)-\eta^2_*(1-\theta_*)}{\eta^2_*(1-\theta_*)\hat\eta^2(1-\hat\theta)}\right|\\
&\le \frac{|\hat\eta^2-\eta^2_*|+\hat\eta^2|\hat\theta-\theta_*|+\theta_*|\hat\eta^2-\eta^2_*|}{\eta^2_*(1-\theta_*)\hat\eta^2(1-\hat\theta)}=O_p(1/\sqrt{n}),
\end{align*}
since $\hat\eta^2,1/\hat\eta^2,1/(1-\hat\theta)=O_p(1)$ and $|\hat\eta^2-\eta^2_*|,|\hat\theta-\theta_*|=O_p(1/\sqrt{n})$.
Similarly, one can check that $|\hat w-w_*|=O_p(1/\sqrt{n})$.
This implies that the $\sqrt{n}$-consistent estimator of $(\eta^2_*,\theta_*)$ guarantees the existence of $\sqrt{n}$-consistent estimator of $\Omega$. The rest of the proof follows the same derivation steps provided in the proofs of Lemma 1 and Theorem 1 of \citet{lee2012simultaneous}. Hence, we omit the proof.
\end{proof}

We define a matrix, $V\in\mathbb{R}^{2\times 2}$, which characterizes the limiting covariance of $(\hat\eta,\hat\theta)$.
This symmetric matrix has the following four elements:
\begin{equation}\label{eq:V.elements}
\begin{split}
&V_{11}=\frac{1}{q^2}\left(\sum_{j=1}^q {\rm var}(E_{ij}^2)+2\sum_{j<k} {\rm cov}(E_{ij}^2,E_{ik}^2)\right),\\  
&V_{22}=\frac{1}{q^2(q-1)^2}\sum_{j,k,l,m}Q_{jk}Q_{lm}{\rm cov}(E_{ij}E_{ik},E_{il}E_{im}),\\
&V_{12}=V_{21}=\frac{1}{q^2}{\rm var}\left(\sum_{j=1}^q E_{ij}^2\right)-\frac{2}{q^2(q-1)}{\rm cov}\left(\sum_{j<k}E_{ij}E_{ik},\sum_{j=1}^q E_{ij}^2\right),   
\end{split}   
\end{equation}
where $Q=(qI_q-1_q1_q^T)$.

\begin{lemma}\label{lem:consist.Omega}
Under the conditions \ref{assm:A1}, \ref{assm:A2} and \ref{assm:A4}, assuming that $\hat B$ is $\sqrt{n}$-consistent to $B_*$, then, for $W=(1,-1/\eta_*^2)^T$, the following holds.
\begin{align*}
\left\|(\hat\eta^2,\hat\theta)-(\eta_*^2,\theta_*)\right\|=O_p(1/\sqrt{n}),
\end{align*}
where $(\hat\eta^2,\hat\theta)$ is a solution to the optimization problem \eqref{main.objective.weight} for fixed $\hat B$.
Furthermore, if $\theta_*\in (0,1)$, then
\begin{align*}
\sqrt{n}(\hat\eta^2,\hat\theta)-(\eta_*^2,\theta_*))\rightarrow_d W^TN_2(0,V),
\end{align*}
where $W=(1,-1/\eta_*^2)^T$, and $V$ is defined in \ref{eq:V.elements}.
\end{lemma}
\begin{proof}
It suffices to assume that true $B_*$ is known, since the natural following step which is the replacement of $B_*$ with $\hat B$ can be directly obtained by application of the derivation used in the proof of Theorem 2 in \citet{lee2012simultaneous}.
Recall the reparametrization of $(\eta^2,\theta)$ to $(\alpha,\gamma)$ via $\alpha=\eta^2(1-\theta)$, $\gamma=\eta^2(1+(q-1)\theta)$ by which we derived the closed form solutions to $\hat \alpha$ and $\hat \gamma$.
The solutions were $\hat\alpha=\frac{qM_1^*-M_2^*}{q(q-1)}$, $\hat\gamma=\max\{\hat\alpha,M_2^*/q\}$, where $M_1^*=\frac{1}{n}\|E\|_F^2$ and $M_2^*=\frac{1}{n}\|E1_q\|^2$ (see \eqref{eq:alpha.form}, \eqref{eq:gamma.form}). 
We define $\hat\delta=\frac{M_1^*}{q}$.
We first find the joint limiting distribution of $(\hat\delta,\hat\alpha)$.
We have
\begin{align*}
\frac{qM_1^*-M_2^*}{q(q-1)}&=\frac{1}{nq(q-1)}\left\{\sum_{i=1}^n\left[\sum_{j=1}^q (q-1)E_{ij}^2-2\sum_{1\le j<k\le q}E_{ij}E_{ik}\right]\right\}\\
&=\frac{1}{nq(q-1)}\sum_{i=1}^n E_i^TQE_i,   
\end{align*}
where $E_i=(E_{i1},\ldots,E_{iq})^T\in\mathbb{R}^q$ and $Q=(qI_q-1_q1_q^T)$. We further have
\begin{align*}
&\mathbb{E}(E_i^TQE_i/(q(q-1)))=\eta_*^2(1-\theta_*),\\
&{\rm var}\left(E_i^TQE_i/(q(q-1))\right)=\frac{1}{q^2(q-1)^2}\sum_{j,k,l,m}Q_{jk}Q_{lm}{\rm cov}(E_{ij}E_{ik},E_{il}E_{im}).
\end{align*}
In addition, $\hat\delta$ can be expressed as $\frac{1}{nq}\sum_{i=1}^n\sum_{j=1}^q E_{ij}^2=\frac{1}{nq} \sum_{i=1}^n E_i^T I_q E_i$, which satisfies
\begin{align*}
&\mathbb{E}(E_i^T I_q E_i/q)=\eta_*^2,\\
&{\rm var}\left(E_i^T I_q E_i/q\right)=\frac{1}{q^2}\left(\sum_{j=1}^q {\rm var}(E_{ij}^2)+2\sum_{j<k} {\rm cov}(E_{ij}^2,E_{ik}^2)\right).
\end{align*}
Likewise, we obtain
\begin{align*}
{\rm cov}\left(E_i^T I_q E_i/q,E_i^TQE_i/(q(q-1))\right)=\frac{1}{q^2}{\rm var}\left(\sum_{j=1}^q E_{ij}^2\right)-\frac{2}{q^2(q-1)^2}{\rm cov}\left(\sum_{j<k}E_{ij}E_{ik},\sum_{j=1}^q E_{ij}^2\right).    
\end{align*}
Thus, we have
\begin{align}\label{delta.alpha.limit}
\sqrt{n}((\hat\delta,\hat\alpha)-(\eta_*^2,\eta_*^2(1-\theta_*)))\rightarrow_d N_2(0,V).
\end{align}
And, since $\hat\eta^2=(1-1/q)\hat\alpha+\hat\gamma/q$, we have
\begin{align*}
\hat\eta^2-\hat\delta=\frac{\hat\gamma}{q}-\frac{M_2^*}{q^2}=\max\left\{\frac{\hat\alpha}{q},\frac{M_2^*}{q^2}\right\}-\frac{M_2^*}{q^2}.
\end{align*}
When $\theta_*\in(0,1)$, $\hat\alpha\rightarrow_p \eta^2(1-\theta)$, and $\max\{\hat\alpha,M_2^*/q\}=M_2^*/q$ with probability tending to $1$, since $M_2^*/q\rightarrow_p\eta_*^2$.
This yields, $\hat\eta^2=\hat\delta$ with probability tending to $1$ and it further implies
\begin{align*}
\sqrt{n}((\hat\eta^2,\hat\alpha)-(\eta^2_*,\eta^2_*(1-\theta_*)))\rightarrow_d N_2(0,V).
\end{align*}
Then, by the same Delta method via $g(x,y)=(x,(x-y)/x)$, we have
\begin{align*}
\sqrt{n}\left(\left(\hat\eta^2,\frac{\hat\eta^2-\hat\alpha}{\hat\eta^2}\right)-(\eta_*^2,\theta_*)\right)\rightarrow_d W^TN_2(0,V),
\end{align*}
where $W=(1,-1/\eta_*^2)^T$, since $\hat\delta>0$ almost surely. 
It is obvious that $\frac{\hat\eta^2-\hat\alpha}{\hat\eta^2}=\frac{\hat\gamma-\hat\alpha}{\hat\gamma+(q-1)\hat\alpha}=\hat\theta$.
On the other hand, when $\theta_*=0$, 
\eqref{delta.alpha.limit} with the Delta method yields,
\begin{align*}
\sqrt{n}\left(\left(\hat\delta,\frac{\hat\delta-\hat\alpha}{\hat\delta}\right)-(\eta_*^2,\theta_*)\right)\rightarrow_d W^TN_2(0,V).
\end{align*}
Then it suffices to show $|\hat\delta-\hat\eta^2|=O_p(1/\sqrt{n})$ and $|\hat\alpha(1/\hat\eta^2-1/\hat\delta)|=O_p(1/\sqrt{n})$.
Indeed, we have
\begin{align*}
q|\hat\delta-\hat\eta^2|\leq \left|\frac{M_2^*}{q}-\hat\alpha\right|=\left|\frac{M_2^*}{q}-\frac{qM_1^*-M_2^*}{q(q-1)}\right|=\left|\frac{M_2^*-M_1^*}{q-1}\right|,
\end{align*}
and $M_2^*-M_1^*=\frac{1}{n}\sum_{i=1}^n E_i^T Q_0 E_i$, where $Q_0$ has 0 for all diagonal elements and $1$ for all off-diagonal elements. 
Thus, by similar asymptotic derivation steps above, one can easily check $|\hat\delta-\hat\eta^2|=O_p(1/\sqrt{n})$.
Since $\hat\eta^2,\hat\delta,\hat\alpha\rightarrow_p \eta_*^2>0$, $|1/\hat\eta^2|,|1/\hat\delta|,|\hat\alpha|=O_p(1)$.
This implies
\begin{align*}
|\hat\alpha(1/\hat\eta^2-1/\hat\delta)|=|\hat\alpha||1/\hat\delta||1/\hat\eta^2||\hat\eta^2-\hat\delta|=O_p(1/\sqrt{n}).  
\end{align*}
Finally, we have
\begin{align*}
\left\|\left(\hat\delta,\frac{\hat\delta-\hat\alpha}{\hat\delta}\right)-\left(\hat\eta^2,\frac{\hat\eta^2-\hat\alpha}{\hat\eta^2}\right)\right\|\le |\hat\eta^2-\hat\delta|+|\hat\alpha(1/\hat\eta^2-1/\hat\delta)|=O_p(1/\sqrt{n}), 
\end{align*}
and it completes the proof.
\end{proof}

Now we give the proof of Theorem \ref{thm:classic.anal}.
\begin{proof}
We already verified that a $\sqrt{n}$-consistent estimator of $(\eta_*^2,\theta_*)$ guarantees the existence of the $\sqrt{n}$-consistent estimator of $\Omega_*$ by Lemma \ref{lem:consist.Beta}.
Combining the result in Lemma \ref{lem:consist.Omega}, the rest of the proof follows the same derivation used in the proofs of Lemma 3 and Theorem 3 of \citet{lee2012simultaneous}.
Hence, we omit the proof.
\end{proof}

\begin{figure}
     \centering
     \begin{subfigure}[b]{0.45\textwidth}
         \centering
         \includegraphics[width=\textwidth]{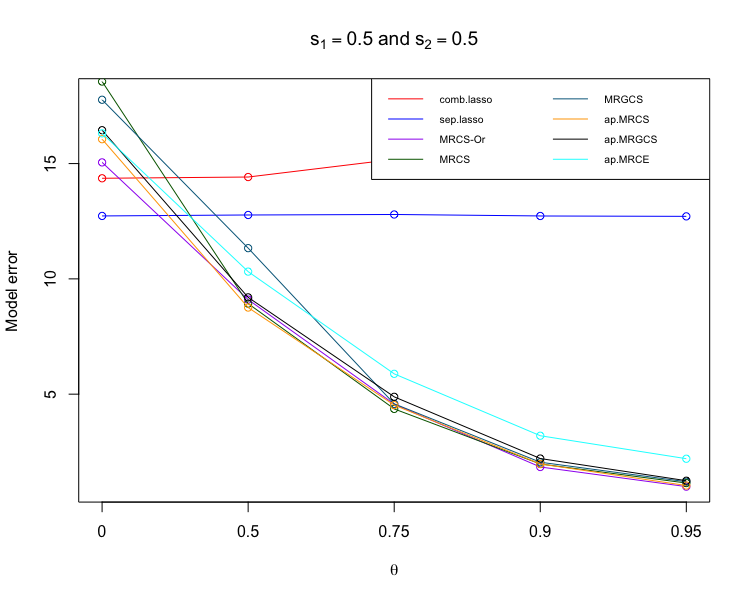}
         \caption{Compound symmetry; $(p,q)=(20,50)$}
         \label{fig:section7-fig1}
     \end{subfigure}
     \hspace{0.2cm}
     \begin{subfigure}[b]{0.45\textwidth}
         \centering
         \includegraphics[width=\textwidth]{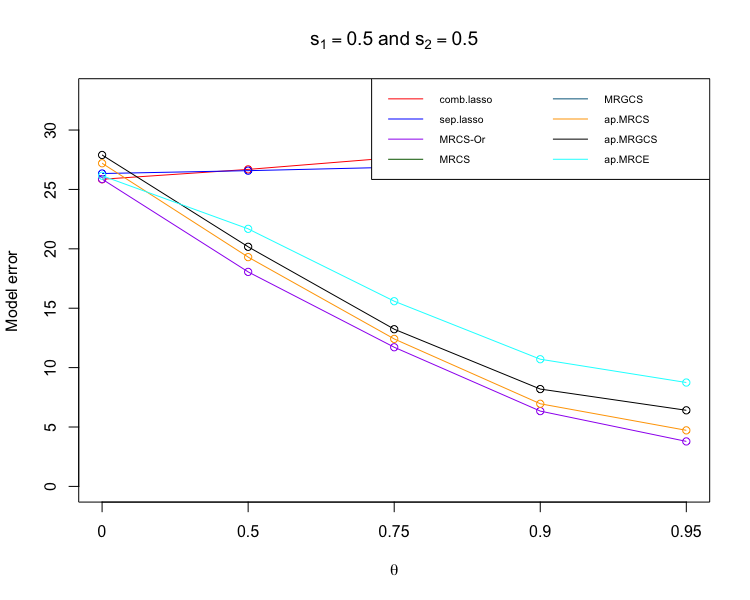}
         \caption{Compound symmetry; $(p,q)=(50,20)$}
         \label{fig:section7-fig2}
     \end{subfigure}\\
     \begin{subfigure}[b]{0.45\textwidth}
         \centering
         \includegraphics[width=\textwidth]{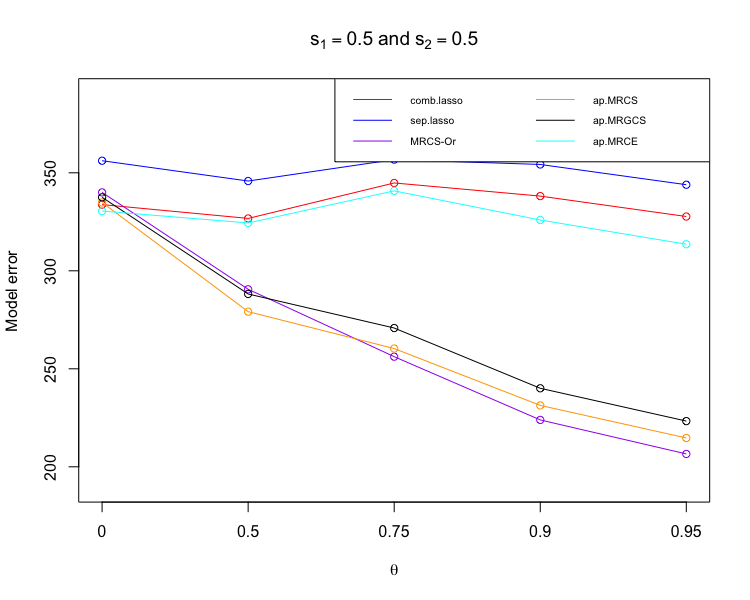}
         \caption{Compound symmetry; $(p,q)=(80,80)$}
         \label{ffig:section7-fig3}
     \end{subfigure}
         \hspace{0.2cm}
    \begin{subfigure}[b]{0.45\textwidth}
         \centering
         \includegraphics[width=\textwidth]{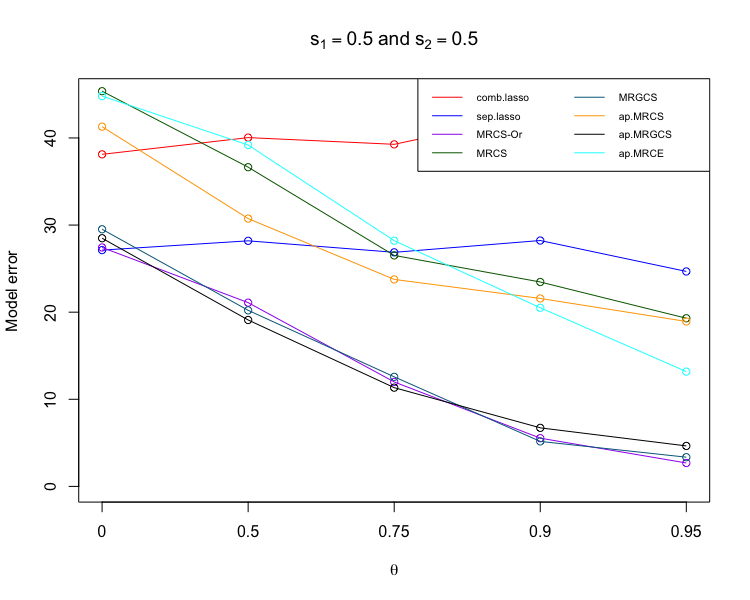}
         \caption{Equicorrelation; $(p,q)=(20,50)$}
         \label{fig:section7-fig4}
     \end{subfigure}\\
    \begin{subfigure}[b]{0.45\textwidth}
         \centering
         \includegraphics[width=\textwidth]{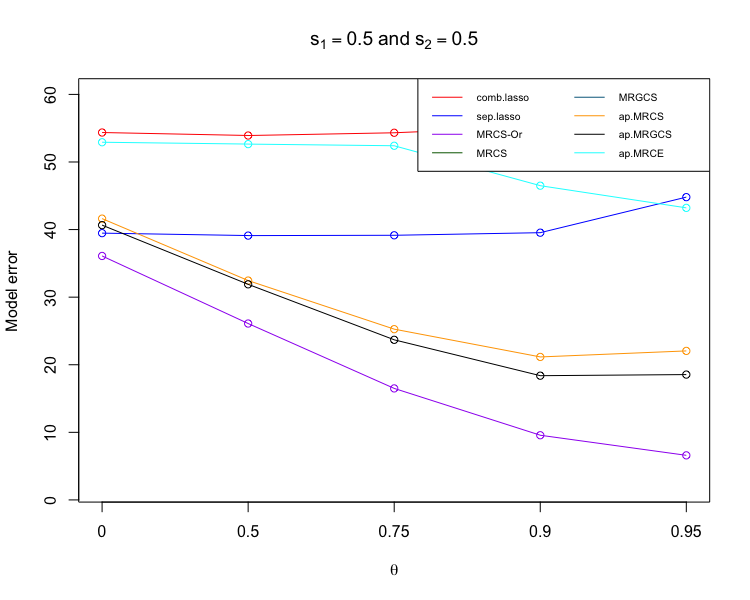}
         \caption{Equicorrelation; $(p,q)=(50,20)$}
         \label{fig:section7-fig5}
     \end{subfigure}
         \hspace{0.2cm}
    \begin{subfigure}[b]{0.45\textwidth}
         \centering
         \includegraphics[width=\textwidth]{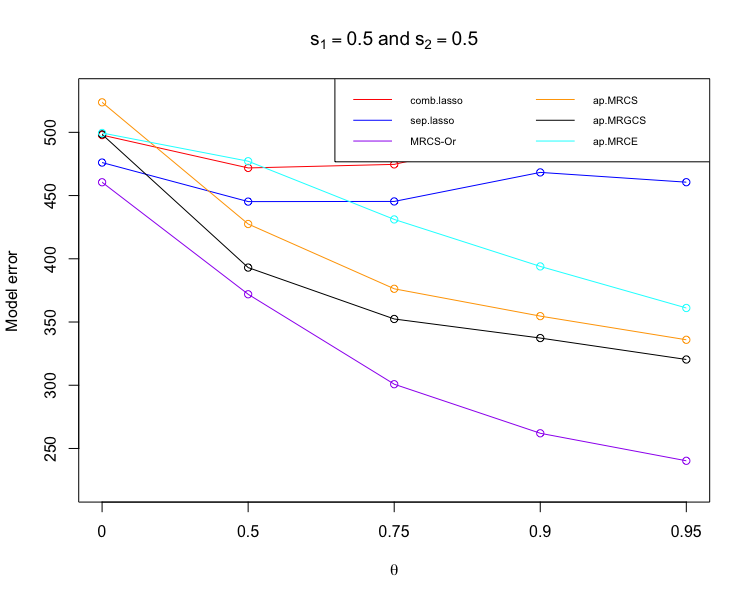}
         \caption{Equicorrelation; $(p,q)=(80,80)$}
         \label{fig:section7-fig6}
     \end{subfigure}\\
     \vspace{-0.2cm}
\caption{Model error comparison plots for the setting $(s_1,s_2) = (0.5,0.5)$ under compound symmetry error covariance and equicorrelation.  The data-generating processes are described in Sections \ref{sec:fixed.eta.simul} and \ref{sec:varied.eta.asym}, respectively.}
\label{fig:section7-6figs}
\end{figure}

\bibliography{allref}

\end{document}